\documentclass[a4paper,12pt]{article}
\usepackage[utf8]{inputenc}
\usepackage[T2A]{fontenc}
\usepackage{indentfirst}
\usepackage{amsmath}
\usepackage{amsthm}
\usepackage{amsfonts}
\usepackage{amssymb}
\usepackage{graphicx}
\usepackage{setspace}
\usepackage{array}
\usepackage{float}
\usepackage{url}
\usepackage{multirow}
\usepackage{multicol}
\usepackage{enumitem}
\usepackage[title]{appendix}
\usepackage{titling}
\usepackage{longtable}
\usepackage{geometry}
\usepackage{comment}
\newtheorem{theorem}{Theorem}
\newtheorem{lemma}{Lemma}

\newtheorem{remark}{Remark}

\newtheorem{corollary}{Corollary}

\newcommand{\fd}[1]{\mathbb{F}_{#1}}
\newcommand{\fdm}[1]{\mathbb{F}_{#1}^{*}}
\newcommand{\fdo}[1]{\mathcal{S}(\mathbb{F}_{#1})}
\newcommand{\finv}[0]{\sigma}

\allowdisplaybreaks
\providecommand{\keywords}[1]{
	\begin{flushleft}
	\small\textbf{Keywords: } #1
	\end{flushleft}
	}
\providecommand{\affili}[1]{
	\begin{center} 
	\small #1
	\end{center}
	\vspace{.1cm}}
\title{On the image of an affine subspace under the inverse function within a finite field}
\author{\normalfont N. Kolomeec, D. Bykov
}
\date{}
\begin{document}
\large
%
%
\maketitle
%
%
\par\vspace{-60pt}
%
%
 \affili{
 	Sobolev Institute of Mathematics, Russia\\
 	\texttt{ kolomeec@math.nsc.ru, den.bykov.2000i@gmail.com }
}
%
%
\begin{abstract}
	\noindent We consider the function $x^{-1}$ that inverses a finite field element $x \in \fd{p^n}$ ($p$ is prime, $0^{-1} = 0$) and affine $\fd{p}$-subspaces of $\fd{p^n}$ such that their images are affine subspaces as well. 
	It is proven that the image of an affine subspace $L$, $|L| > 2$, is an affine subspace if and only if $L = q \fd{p^k}$, where $q \in \fdm{p^n}$ and $k \mid n$. In other words, it is either a subfield of $\fd{p^n}$ or a subspace consisting of all elements of a subfield multiplied by $q$. This generalizes the results that were obtained for linear invariant subspaces in~2006. 
	As a consequence, we propose a sufficient condition providing that a function $A(x^{-1}) + b$ has no invariant affine subspaces $U$ of cardinality $2 < |U| < p^n$ for an invertible linear transformation $A: \fd{p^n} \to \fd{p^n}$ and $b \in \fdm{p^n}$. As an example, it is shown that the condition works for S-box of AES.
	Also, we demonstrate that some functions of the form $\alpha x^{-1} + b$  have no invariant affine subspaces except for $\fd{p^n}$, where $\alpha, b \in \fdm{p^n}$ and $n$ is arbitrary. 
\end{abstract}
%
%
%
%

\section{Introduction}
	 
	In this work we consider affine subspaces whose images under the function that inverses an element of a finite field are affine subspaces as well. 
	We denote this function by $\finv$, i.е. $\finv: \fd{p^n} \to \fd{p^n}$ such that 
	$$
		\finv(x) = \begin{cases}
			x^{-1}, & \text{if }  x \neq 0\\
			0, & \text{if }  x = 0
		\end{cases},
	$$    
	where $\fd{p^n}$ is \emph{a finite field} consisting of $p^n$ elements. Hereinafter $p$ is a prime number.
	\emph{A linear subspace} $L$ of $\fd{p^n}$ is its linear $\fd{p}$-subspace. In other words, $L$ is an additive subgroup of $\fd{p^n}$.
	We denote an affine subspace of $\fd{p^n}$ by $a + L$. For any $S \subseteq \fd{p^n}$, $a \in \fd{p^n}$ and $f: \fd{p^n} \to \fd{p^n}$ we define 
	 $$
	 	a + S = \{ a + s : s \in S\}, \  aS = \{as : s \in S\} \text{ and } f(S) = \{f(s) : s \in S\}.
	 $$
	It is known that $\fd{p^k}$ is a subfield of $\fd{p^n}$ if and only if $ k \mid n$. In this case we assume that $\fd{p^k} \subseteq \fd{p^n}$. Also, $\fdm{p^n} = \fd{p^n} \setminus \{0\}$.
	According to the denotations above, we consider affine subspaces $U$ of $\fd{p^n}$ such that $\finv(U)$ are affine subspaces as well. We are interested in \emph{nontrivial} subspaces that satisfy $2 < |U| < p^n$ since any set $S$ of cardinality $1$, $2$ (for $p = 2$) or $p^n$ is always an affine subspace. 
	
	These subspaces are connected with invariant subspaces. If we know all $U$, we can construct all invariant affine subspaces for a function that is affinely equivalent to $\finv$.
	Recall that $S \subseteq \fd{p^n}$ is \emph{an invariant set} of $f$ if $f(S) \subseteq S$. Functions $f$ and $g$ are \emph{affinely equivalent} if and only if $g(x) = A(f(B(x) + b)) + a$ for all $x \in \fd{p^n}$, where $A$ and $B$ are invertible linear transformations, $a, b \in \fd{p^n}$. Similarly to subspaces, we consider linear transformations assuming that $\fd{p^n}$ is a vector space over $\fd{p}$.

	The function $\finv$ is interesting for cryptography. For instance, S-box of AES~\cite{AesFIPS, AesBook} is based on $\finv$. Its invariant linear subspaces were used to investigate group properties of the AES-like round functions~\cite{CarantiEtAl2006,CarantiEtAl2009a,CarantiEtAl2009b}. 
	An attack using an invariant subspace of several rounds of ciphering was introduced in~\cite{LeanderEtAl2011}. A more general attack which can use an invariant set is also known~\cite{TodoEtAl2016}. 
	Note that invariant sets for both of them may be constructed starting with S-boxes, see, for instance, \cite{Burov2021,TrifonovFomin2021}. 
	Note that D. Goldstein et al. and S. Mattarei proved that all invariant linear subspaces of the function $\finv$ are subfields of $\fd{p^n}$, see~\cite{GoldsteinEtAl2006,Mattarei2007}. Thus, we generalize this result in two directions: for affine subspaces and for functions that are affinely equivalent to $\finv$.
	
	At the same time, if $f(U)$ is not an affine subspace for any $2$-dimensional affine subspace $U$ of $\fd{2^n}$, then $f$ is an APN function~\cite{Nyberg1994}. Moreover, only APN functions satisfy this property.
	They are very important for cryptography. Also, there are many open questions in this area, see, for instance, \cite{Carlet2015}. If $p = 2$ and $n$ is odd, then $\finv$ is an APN function.

		

    Our contribution is the following. In Section~\ref{section:AllAffineSubspaces}, we prove that both $U$ and $\finv(U)$, $|U| > 2$, are affine subspaces if and only if $U = q \fd{p^k}$, where $q \in \fdm{p^n}$ and $k \mid n$ (Theorem~\ref{th:stableSubspaces}). Next, we show how to apply this to construct a function of the form $A(\finv(x)) + b$ without invariant affine subspaces, where $A$ is an invertible linear transformation and $b \in \fdm{p^n}$. Section~\ref{section:InvariantSubspaces} contains a sufficient condition for such functions to have no nontrivial invariant affine subspaces (Theorem~\ref{th:invatiantSuffCondition}). 
    Also, for the function $\alpha \finv(x) + b$, where  $\alpha, b \in \fdm{p^n}$, we provide a criterion on the existence of invariant affine subspaces excluding the whole $\fd{p^n}$ (Theorem~\ref{th:invAlphaB}). Section~\ref{section:AES} demonstrates that the results obtained can theoretically prove that there are no nontrivial invariant affine subspaces for S-box of AES.

\section{Subspaces that are mapped to subspaces}\label{section:AllAffineSubspaces}

In this section we find all nontrivial affine subspaces $U$ such that $\finv(U)$ is an affine subspace as well. 
We start with the following well known property of linear subspaces.
\begin{lemma}\label{lemma:linersum0} Let $L$ be a linear subspace of $\fd{p^n}$ and $|L| > 2$. Then $\sum_{x \in L} x = 0$.
\end{lemma}
\begin{proof} 
Let $s = \sum_{x \in L} x$, $s \in \fd{p^n}$. Let $p > 2$. Since $(p-1)L = L$, it holds 
$$
	s = \sum_{x \in L}x  = \sum_{x \in L} (p-1)x = (p-1)s.
$$
But $p - 1 \neq 1$. This implies that $s = 0$.

Let $p = 2$. Then $\dim L \geq 2$ and $L = L' \cup (a + L')$, where $L'$ is a $(\dim L - 1)$-dimensional linear subspace of $L$, $a \in L$. Therefore,
$$
	s = \sum_{x \in L}x  = \sum_{x \in L'} x + \sum_{x \in L'} (a + x) = 2^{\dim L - 1} a = 0.
$$
The lemma is proven.
\end{proof}

Actually, it is enough to consider only linear subspaces.
\begin{theorem}\label{th:noaffine} Let $U$ be an affine subspace of $\fd{p^n}$, $|U| > 2$ and $\finv(U)$ be an affine subspace as well. Then $U$ is a linear subspace.
\end{theorem}
\begin{proof} 
Let $U = a + L$ and $\finv(U) = U' = a' + L'$, where $L$ and $L'$ are linear subspaces, $a,a' \in \fd{p^n}$. 

First of all, we prove some property of $U$ and $U'$. Taking into account that $|U| > 2$, Lemma~\ref{lemma:linersum0} provides that 
$$
	\sum_{y \in \finv(U)}y = \sum_{y \in L'} (a' + y) = p^{\dim U} a' + \sum_{y \in L'} y = 0 + 0 = 0.
$$
Since $\finv(U) = \{ \finv(x + s) : s \in L\}$ for any $x \in U$, we obtain that 
\begin{align}\label{eq:solutionsFromU}
	\sum_{s \in L} \finv(x + s) = 0 \text{ for any } x \in U.
\end{align}
Let us define the following polynomial $h$ over $\fd{p^n}$:
\begin{align*}
	h(y) = \sum_{s \in L} \prod_{t \in L\setminus\{s\}}(y + t).
\end{align*}

Now we prove by contradiction that $U$ is a linear subspace. Suppose that $U$ is not a linear subspace. Let $x$ be any element of $U$. Therefore, $x \neq s$ for any $s \in L$. At the same time, $x + s \neq 0$ since $-s \in L$ as well. As a result, 
\begin{align*}
	h(x) = \sum_{s \in L} \prod_{t \in L\setminus\{s\}}(x + t) = \prod_{t \in L}(x + t) \sum_{s \in L} \finv(x + s) \stackrel{(\ref{eq:solutionsFromU})}{=} 0.
\end{align*}
Thus, any $x \in U$ is a root of $h$. However, the degree of $h$ is less than $|U| = |L|$. Also, $h \not\equiv 0$ since 
$$
	h(0) = \prod_{t \in L\setminus\{0\}}(0 + t) = \prod_{t \in L\setminus\{0\}}t \neq 0.
$$
In other words, $h$ cannot have $|U|$ roots. This is a contradiction. Hence, $U$ is a linear subspace.
The theorem is proven.
\end{proof}

The rest of the proofs of this section use Hua's identity~\cite{Hua1949} (it can also be found in~\cite{Jacobson1974,LidlNiederreiter1983}). They are very similar to ones proposed in~\cite{GoldsteinEtAl2006,Mattarei2007} for linear invariant subspaces of the function $\finv$. 
\begin{lemma}\label{lemma:a2b}
	Let $L$ and $\finv(L)$ be linear subspaces of $\fd{p^n}$. Then $a^2 \finv(b) \in L$ for any $a, b \in L$. 
\end{lemma}
\begin{proof} 
Hua's identity
$$
	a - (a^{-1} + (b'^{-1} - a)^{-1})^{-1} = ab'a
$$
holds for any $a, b' \in \fd{p^n}$, $a,b' \neq 0, ab' \neq 1$. Note that it works not only for fields and can be easily proven. 

Let $b' = \finv(b)$.
It is clear that $a^2 \finv(b) \in L$ if $a = 0$, $b' = 0$ or $a b' = 1$. Next, we consider other cases. It means that $b' = b^{-1}$. Since $\finv(\finv(L)) = L$  and both $L$ and $\finv(L)$ are linear subspaces, the identity provides that 
\begin{align*}
b - a = b'^{-1} - a &\in L,\\ 
a^{-1} + (b - a)^{-1} &\in \finv(L),\\ 
a - (a^{-1} + (b - a)^{-1})^{-1} &\in L. 
\end{align*}
Hence, $ab'a = a^2 \finv(b) \in L$ as well.
The lemma is proven.
\end{proof}

The following theorem describes all considered subspaces.
\begin{theorem}\label{th:stableSubspaces}
	Let $U$ be an affine subspace of $\fd{p^n}$ and $|U| > 2$. Then $\finv(U)$ is an affine subspace if and only if $U = q \fd{p^k}$, where $k \mid n$ and $q \in \fdm{p^n}$.
\end{theorem}
\begin{proof} 
First of all, $\finv(\fd{p^k}) = \fd{p^k}$ for $k \mid n$ since it is a subfield of $\fd{p^n}$. Also, $\finv(q \fd{p^k}) = q^{-1}\finv(\fd{p^k})$ which is a linear subspace. The sufficient condition is proven. 

Next, we prove the necessary condition. Let $L = U$ and $\finv(L)$ be an affine subspace. By Theorem~\ref{th:noaffine}, $L$ is a linear subspace. In fact, $\finv(L)$ is linear as well since $0 \in L$ and $\finv(0) = 0$. Without loss of generality we assume that $1 \in L$ since we can consider any $q L$ instead of $L$, $q \in \fdm{p^n}$. Indeed, $\finv(qL)$ is a linear subspace since $\finv(qL) = q^{-1} \finv(L)$. In light of that, we need only to prove that $L$ is a subfield of $\fd{p^n}$.

According to Lemma~\ref{lemma:a2b}, $a^2 \finv(b) \in L$ for any $a, b \in L$. But $1 \in L$. It means that $\finv(b) \in L$ and $a^2 \in L$. As a consequence, $L = \finv(L)$ and $a^2 b, a^2 \in L$ for any $a, b \in L$.

We know that $0, 1 \in L$, $L$ is an additive subgroup and all inversions belong to $L$. The final step is to prove that  $bc \in L$ for any $b, c \in L$. There are two cases.

\textbf{Case 1.} $p = 2$. Let $b, c \in L$. Since  $a^2 \in L$ for any $a \in L$, there always exists $a \in L$ such that $c = a^2$. Thus, $bc = a^2b \in L$. 

\textbf{Case 2.} $p \neq 2$. For any $b, c \in L$ it holds that $2bc = (c + 1)^2 b - b - c^2b \in L$. Indeed, all these terms belong to $L$. Since $2 \in \fd{p}$ and $L$ is a linear subspace, $bc \in L$ as well.

Thus, $L$ is a subfield of $\fd{p^n}$, i.e. $L = \fd{p^k}$, where $k \mid n$. The theorem is proven.
\end{proof}

\begin{remark}
    Theorem~\ref{th:stableSubspaces} shows that for prime $n > 2$ and $p = 2$ the APN function $\finv$ maps any nontrivial affine subspace to a set that is not an affine subspace.
\end{remark}

\section{Functions without invariant subspaces}\label{section:InvariantSubspaces}

 Here we show how to apply the results of Section~\ref{section:AllAffineSubspaces} to construct  functions without invariant affine subspaces of the form
 $$
    \finv_{A,b}(x) = A(\finv(x)) + b,
 $$ 
 where $A : \fd{p^n} \to \fd{p^n}$ is an invertible linear transformation and $b \in \fd{p^n}$. They are affinely equivalent to $\finv$.
Let us define
$$
	\fdo{p^n} = \{ x \in \fd{p^n} : x \notin \fd{p^k}, \text{ where } k \mid n \text{ and } k < n\}.
$$
In other words, $\fdo{p^n}$ contains all elements of $\fd{p^n}$ which do not lie in its subfields. It is always nonempty.
\begin{theorem}\label{th:invatiantSuffCondition}
	Let $A: \fd{p^n} \to \fd{p^n}$ be an invertible linear transformation, $b \in \fdm{p^n}$ and
\begin{equation}\label{eq:ABcondition}
	b^{-1}A(b^{-1}) \in \fdo{p^n}.
\end{equation}
Then $\finv_{A,b}$ has no invariant affine subspaces $U$ such that $2 < |U| < p^n$.
\end{theorem}
\begin{proof}
Let $U$ be affine subspace of $\fd{p^n}$ and $\finv_{A,b}(U) = U$, where $|U| > 2$.  Theorem~\ref{th:stableSubspaces} provides that $U = q^{-1}K$ and $b + A(qK) = U$ for some subfield $K$ of $\fd{p^n}$ and $q \in \fdm{p^n}$. This is equivalent to
\begin{equation}\label{eq:invK}
	 qb + qA(qK) = K.
\end{equation}
Since $0 \in K$, (\ref{eq:invK}) provides that $qb = qb + qA(q \cdot 0)  \in K$. Thus, $q = b^{-1} u$, where $u \in K \setminus \{0\}$. According to~(\ref{eq:invK}), 
$$
	qb + qA(q \cdot u^{-1}) = qb + u b^{-1} A(b^{-1} u u^{-1}) = qb + u b^{-1} A(b^{-1}) \in K.
$$
At the same time, both $qb$ and $u$ belong to $K$. Hence, $b^{-1}A(b^{-1}) \in K$. As a result, $b^{-1}A(b^{-1}) \notin \bigcup_{k \mid n, k < n} \fd{p^k}$ implies that $U = K = \fd{p^n}$.
\end{proof}
\begin{remark}\label{remark:invatiantSuffCondition}
	The condition~(\ref{eq:ABcondition}) is equivalent to
	\begin{equation}\label{eq:InvABcondition}
		b^{-1}\finv_{A,b}(b) \in \fdo{p^n}.
\end{equation}
\end{remark}
\begin{proof}
	Indeed, $b^{-1}\finv_{A,b}(b) = b^{-1}A(b^{-1}) + 1$. At the same time, $x \in \fdo{p^n}$ if and only if $x + 1\in \fdo{p^n}$ since $1$ lies in any subfield of $\fd{p^n}$.    
\end{proof}
\begin{remark}
Recall that it is not difficult to determine if $x \in \fdo{p^n}$. It is enough to check that $x^{p^k} \neq x$ for any $k \mid n$, $k < n$.
\end{remark}

Even if $A$ does not satisfy~(\ref{eq:ABcondition}), we can modify it and obtain a function without nontrivial invariant subspaces.
\begin{corollary}
	Let $A: \fd{p^n} \to \fd{p^n}$ be invertible linear transformation, $b \in \fdm{p^n}$. Then the transformation $A' : \fd{p^n} \to \fd{p^n}$ such that
$$
	A'(x) = \alpha \beta A(x), \text{ where } \alpha \in \fdo{p^n} \text{ and } \beta = (b^{-1}A(b^{-1}))^{-1},
$$
together with $b$ satisfy the conditions of Theorem~\ref{th:invatiantSuffCondition}.
\end{corollary}
\begin{proof}
$A'$ is defined correctly since $b^{-1}A(b^{-1}) \neq 0$. Indeed, $A(0) = 0$ and $A$ is invertible. It means that $A(b^{-1}) \neq 0$.
Let us check the properties of $A'$. It is clear that $A'$ is linear and invertible since it is a composition of two linear invertible transformations.  Also, let us check~(\ref{eq:ABcondition}):
\begin{align*}
	b^{-1}A'(b^{-1}) &= b^{-1} \alpha (b^{-1}A(b^{-1}))^{-1} A(b^{-1}) = \alpha.
\end{align*}
As a result, $A'$ and $b$ satisfy the conditions of Theorem~\ref{th:invatiantSuffCondition}. 
\end{proof}

However, the proposed condition does not guarantee that $\finv_{A,b}$ has no invariant affine subspaces $U$ such that $1 \leq |U| \leq 2$, see, for instance, the next section. It is clear that
\begin{itemize}
 	\item any $U = \{ u \}  \subseteq \fd{p^n}$ is invariant for $\finv_{A,b}$ if and only if $u$ is a fixed point of $\finv_{A,b}$, 
	\item any $U = \{ u, v\} \subseteq \fd{p^n}$, where $u \neq v$, is an invariant affine subspace if and only if $p = 2$ and either $u,v$ are fixed point of $\finv_{A,b}$ or $u = \finv_{A,b}(v)$, $v = \finv_{A,b}(u)$.
\end{itemize}
It is not easy to exclude these cases for arbitrary $A$. At the same time, we can do this for the simplest linear functions. 
\begin{theorem}\label{th:invAlphaB}
	Let  $\alpha, b \in \fdm{p^n}$, $\finv_{\alpha, b} : \fd{p^n} \to \fd{p^n}$ such that $$\finv_{\alpha,b}(x) = \alpha\finv(x) + b.$$ Then  $\finv_{\alpha, b}$ has no invariant affine subspaces except for $\fd{p^n}$ if and only if
	\begin{itemize}
		\item $\alpha b^{-2} \in M_2$ for $p = 2$, where
	\begin{align}\label{eq:M2}
		M_2 &= \{ x \in \fdo{2^n} : \mathrm{tr}(x) = 1 \}, \mathrm{tr}(x) = x^{2^0} + x^{2^1} + \ldots + x^{2^{n-1}},
	\end{align}	
		\item $\alpha b^{-2} + 4^{-1} \in M_p$ for $p \neq 2$, where
	\begin{align}\label{eq:Mp}
		M_p &=	\{ x \in \fdo{p^n} : x \neq y^2 \text{ for any } y \in \fd{p^n} \}.
	\end{align}	
	\end{itemize}
\end{theorem} 
\begin{proof}
	Let $U$ be an invariant affine subspace of $\finv_{\alpha, b}$. We consider a criterion of the existence of such $U$ depending on $|U|$.
		
	\textbf{Case 1.} $2 < |U| < p^n$, i.e. $U$ is nontrivial. According to Theorem~\ref{th:invatiantSuffCondition}, $\alpha b^{-2} \in \fdo{p^n}$ implies that such $U$ does not exist. Let us prove that such $U$ exists if $\alpha b^{-2} \notin \fdo{p^n}$, i.e. $\alpha b^{-2} \in K$ for some subfield $K$ of $\fd{p^n}$, $|K| < |p^n|$. According to~(\ref{eq:invK}), if
	\begin{align*}
		qb + q \alpha q x \in K \text{ for any } x \in K,
	\end{align*}	
	then $U = qK$ is invariant. Let $q = b^{-1}$. Then, $1 + (\alpha b^{-2}) x \in K$ for any $x \in K$. In other words, $U = b^{-1} K$ is invariant.
	
	Thus, such $U$ does not exist if and only if $\alpha b^{-2} \in \fdo{p^n}$. Also, if $p \neq 2$, we can see that $\alpha b^{-2} \in \fdo{p^n}$ if and only if $\alpha b^{-2} + 4^{-1} \in \fdo{p^n}$ since $4^{-1}$ lies in any subfield of $\fd{p^n}$.

	\textbf{Case 2.} $|U| = 1$, i.e. we consider fixed points of $\finv_{\alpha, b}$. Let $x \in \fdm{p^n}$ since $0$ is not a fixed point of $\finv_{\alpha, b}$. Then $x$ is a fixed point if and only if
	\begin{align*}
		\alpha x^{-1} + b = x.
	\end{align*}
	Multiplying both parts by $xb^{-2} \neq 0$, we obtain the following equivalent equation:
    \begin{align}\label{eq:fixedPoint}
		y^2 - y - \alpha b^{-2} = 0, \text{ where } y = xb^{-1}.
	\end{align}
	
	\textbf{Case 2.1.} $p = 2$. 	It is known that the equation~(\ref{eq:fixedPoint}) has a solution if and only if $\mathrm{tr}(\alpha b^{-2}) = 0$, see, for instance,~\cite{CherlyEtAl1998}.

	\textbf{Case 2.2.} $p \neq 2$. The equation~(\ref{eq:fixedPoint}) has a solution if and only if $(y - 2^{-1})^2 = 4^{-1} + \alpha b^{-2}$ has a solution. 
	Thus, there is a fixed point if and only if $\alpha b^{-2} + 4^{-1} \in \{ x^2 : x \in \fd{p^n}\}$.
	
	\textbf{Case 3.} $|U| = 2$, i.e. $U = \{x, y\} \subseteq \fd{p^n}$, $x \neq y$. Let $U$ be invariant. Then 
	\begin{align}
		\alpha \finv(\alpha \finv(x) + b) + b = x, \text{ i.e. }\nonumber\\
		\alpha \finv(x) + b  = \alpha \finv(x - b). \label{eq:invDim2Condition}
	\end{align}
	
	\textbf{Case 3.1.} $x \in \{0, b\}$. If $x = 0$, then $\alpha \finv(0) + b = b \neq 0$. Let $x = b$. According to~(\ref{eq:invDim2Condition}), $\alpha \finv(b) + b = 0 \neq b$. In other word, $U = \{0, b\}$ for both cases. Also, $\alpha \finv(b) + b = 0$ is equivalent to $\alpha b^{-2} = -1$. Moreover, it is not difficult to see that $\finv_{\alpha, b}(\{0, b\}) = \{0, b\}$ if and only if $\alpha b^{-2} = -1$. At the same time, $-1 \notin \fdo{p^n}$.
	
	\textbf{Case 3.2.} $x \notin \{0, b \}$. We multiply both parts of the last equation by $b^{-2}x(x - b)$:
	\begin{align*}
		b^{-2}x(\alpha &\finv(x) + b) (x - b)  = b^{-2} x\alpha, \text{ i.e. }\\
		b^{-2}(bx^2 - \alpha b - b^2 &x) = b^{-2}(\alpha - \alpha \finv(x) x) x, \text{ or, equivalently,}\\
		&(xb^{-1})^2 - xb^{-1} - \alpha b^{-2}  = 0.
	\end{align*}
	Since it is equivalent to~(\ref{eq:fixedPoint}), $x$ is a fixed point of $\finv_{\alpha, b}$, see case 2.

	Finally, $\fd{p^n}$ is the only invariant affine subspace of $\finv_{\alpha, b}$ if and only if $\alpha b^{-2} \in M_2$ for $p = 2$ and $\alpha b^{-2} + 4^{-1} \in M_p$ for $p \neq 2$.
\end{proof}
\begin{remark}
    The sets $M_p$ defined in~(\ref{eq:M2}) and (\ref{eq:Mp}) are nonempty.
\end{remark}
\begin{proof}
    It is straightforward that $|\{ x \in \fd{2^n} : \mathrm{tr}(x) = 1\}| = 2^{n-1}$ and $|\{ x \in \fd{p^n} : x \neq y^2 \text{ for any } y \in \fd{p^n}\}| = \frac{1}{2}(p^{n} - 1)$.
    At the same time, we can estimate the number of elements that lie in subfields of $\fd{p^n}$:
$$
	|\bigcup\limits_{k|n,k < n} \fd{p^k}| < 1 + p + p^2 + \ldots + p^{n - 2} = \frac{p^{n - 1} - 1}{p - 1}, \text{ where } n \geq 3.
$$
This implies that $M_p$ is nonempty for $n \geq 3$. It is not difficult to check that $M_p$ is nonempty for $n = 1,2$ as well.
\end{proof}


Theorem~\ref{th:invAlphaB} allows us to construct a permutation of an arbitrary $\fd{p^n}$ without invariant affine subspaces (except for the whole $\fd{p^n}$). 
However, its algebraic structure is very simple. Also, if $n$ is not prime, there always exists a linear subspace $L \neq \fd{p^n}$ such that $\finv_{\alpha, b}(u + L) = v + L$ for some $u,v \in \fd{p^n}$. This makes the function $\finv_{\alpha, b}$ not very good for cryptography.

\section{S-box of AES}\label{section:AES}

The Advanced Encryption Standard (AES) is the most popular block cipher
whose S-box $S : \fd{2}^8 \to \fd{2}^8$ (see~\cite[Figure 7]{AesFIPS}) is based on the function $\finv_{A,b}$ for certain invertible linear transformation $A$ and $b \in \fd{2^8}$, see~\cite{AesFIPS,AesBook}. The field $\fd{2^8}$ is constructed as the quotient ring of polynomials modulo
$$
	m(x) = x^8 + x^4 + x^3 + x + 1 \text{ over } \mathbb{F}_2
$$
and $b = x^6 + x^5 + x + 1 = \mathtt{0x63}$. Let us consider its invariant subspaces.

\begin{itemize}
	\item \textbf{Nontrivial affine subspaces.} It can be shown that $t = b^{-1}S(b) \in \fdo{2^8}$. Indeed, $S(b) = x^7 + x^6 + x^5 + x^4 + x^3 + x + 1 = \mathtt{0xFB}$. It is not difficult to check that $t = x^7 + x^6 + x^3 = \mathtt{0xC8}$.
	Also, 
	\begin{align*}
	    t^2 &= x^6 + x^5 + x^4  + 1 \neq t,\\
	    t^{2^2} &= x^7 + x^6 + x^4 + x^3 + x^2 + 1 \neq t,\\
	    t^{2^4} &= x^7 + x^4 + x^3 + 1 \neq t.
	\end{align*}
	Hence, $t \in \fdo{2^8}$. In other words, Theorem~\ref{th:invatiantSuffCondition} and  Remark~\ref{remark:invatiantSuffCondition} guarantee that $S$ has no nontrivial invariant affine subspaces.
	\item \textbf{Fixed points.} It is known that $S(x) \neq x$ for any $x \in \mathbb{F}_2^8$ by the choice, see~\cite{AesBook}. Therefore, $S$ has no invariant affine subspaces of dimension $0$.
	\item However, 
	$$
		S(\mathtt{0x73}) = \mathtt{0x8F} \text{ and } S(\mathtt{0x8F}) = \mathtt{0x73}.
	$$
	This implies that
	$$
	    \{\mathtt{0x73},\mathtt{0x8F}\} = \{ x^6 + x^5 + x^4 + x + 1, x^7 + x^3 + x^2 + x + 1\}$$ 
	 is its invariant affine subspace of dimension $1$. 
\end{itemize}

Thus, the results obtained allow us to prove theoretically that $S$ has no nontrivial invariant affine subspaces. At the same time, it has an invariant affine subspace of dimension $1$.


\bigskip
\noindent {\normalsize \textbf{Acknowledgements} The work is supported by Mathematical Center in Akademgorodok under agreement No. 075--15--2022--281 with the Ministry of Science and Higher Education of the Russian Federation. }

\end{document}